\documentclass[11pt]{article}

\usepackage{lmodern}
\usepackage[utf8]{inputenc}
\usepackage{amsmath,amssymb,mathrsfs,mathtools}
\usepackage{enumitem}
\usepackage[dvipsnames]{xcolor}
\usepackage{tikz}
\usetikzlibrary{positioning}
\usepackage[margin=1in]{geometry}
\usepackage{microtype}
\usepackage[colorlinks=true,urlcolor=Blue,citecolor=Green,linkcolor=BrickRed]{hyperref}
\usepackage{aliascnt}
\usepackage{cleveref}
\allowdisplaybreaks

\usepackage{yhmath}

\makeatletter
\def\@begintheorem#1#2{\trivlist
   \item[\hskip\labelsep{\bfseries #1\ #2.}]\itshape}
\def\@opargbegintheorem#1#2#3{\trivlist
   \item[\hskip\labelsep{\bfseries #1\ #2\ (#3).}]\itshape}
\makeatother

\newenvironment{proof}[1][Proof]{\par\noindent\textbf{#1.}\quad}{\hfill$\square$\par\medskip}

\newcommand{\kw}[1]{\textbf{#1}}
\newcommand{\algin}{\hspace*{1.4em}}
\newcommand{\algcmt}[1]{{\footnotesize\color{gray}$\vartriangleright$ #1}}
\newcommand{\algline}[3]{#1 & #2 & #3\\}

\newenvironment{algo}
{\par\medskip\noindent\begin{minipage}{\linewidth}
 \hrule height .45pt
 \vspace{.45em}
 \small}
{\vspace{.2em}
 \hrule height .45pt
 \end{minipage}
 \par\medskip}

\newcommand{\thup}[0]{\mathbin {\begin{tikzpicture}[scale=0.1]  \draw[->,thick] (0, 1.5 ) -- (2,1.5);     \draw[->,thick] (0,1.5) -- (0,3.5);\end{tikzpicture} } }
\newcommand{\thdown}[0]{\mathbin{\begin{tikzpicture}[scale=0.1]  \draw[->,thick] (0, 1.5 ) -- (2,1.5);     \draw[->,thick] (0,1.5) -- (0,-0.5);\end{tikzpicture} } }
\newcommand{\mitsuup}[0]{\forall 3\hspace{2pt} \thup}
\newcommand{\mitsudown}[0]{\forall 3\hspace{2pt} \thdown}

\newcommand{\eo}{\mathrm{EO}}

\newcommand{\ceo}{\#\mathrm{EO}}

\newcommand{\pnp}{\mathrm{FP}^{\mathrm{NP}}}
\newcommand{\eoe}{\mathrm{HW}^{=}}

\newcommand{\eosg}{\mathrm{HW}^{>}}
\newcommand{\eosl}{\mathrm{HW}^{<}}
\newcommand{\su}{\mathcal{S}}
\newcommand{\eom}[1][\mathrm{M}]{\mathrm{EO}^{#1}}
\newcommand{\Ftwo}{\mathbb{F}_2}
\newcommand{\Aclass}{\mathscr{A}}
\newcommand{\Pclass}{\mathscr{P}}

\newcommand{\POLMID}{\mathsf{Pol}(\oplus_3)}
\newcommand{\POLUP}{\mathsf{Pol}^\uparrow(\oplus_3)}
\newcommand{\POLDOWN}{\mathsf{Pol}^\downarrow(\oplus_3)}

\DeclareMathOperator{\wt}{wt}

\newtheorem{theorem}{Theorem}[section]

\newaliascnt{lemma}{theorem}
\newtheorem{lemma}[lemma]{Lemma}
\aliascntresetthe{lemma}

\newaliascnt{proposition}{theorem}
\newtheorem{proposition}[proposition]{Proposition}
\aliascntresetthe{proposition}

\newaliascnt{corollary}{theorem}
\newtheorem{corollary}[corollary]{Corollary}
\aliascntresetthe{corollary}

\newaliascnt{definition}{theorem}
\newtheorem{definition}[definition]{Definition}
\aliascntresetthe{definition}

\newaliascnt{remark}{theorem}
\newtheorem{remark}[remark]{Remark}
\aliascntresetthe{remark}

\crefname{theorem}{Theorem}{Theorems}
\Crefname{theorem}{Theorem}{Theorems}

\crefname{lemma}{Lemma}{Lemmas}
\Crefname{lemma}{Lemma}{Lemmas}

\crefname{proposition}{Proposition}{Propositions}
\Crefname{proposition}{Proposition}{Propositions}

\crefname{corollary}{Corollary}{Corollaries}
\Crefname{corollary}{Corollary}{Corollaries}

\crefname{definition}{Definition}{Definitions}
\Crefname{definition}{Definition}{Definitions}

\crefname{remark}{Remark}{Remarks}
\Crefname{remark}{Remark}{Remarks}

\newcounter{myalgorithm}
\renewcommand{\themyalgorithm}{\arabic{myalgorithm}}

\crefname{myalgorithm}{Algorithm}{Algorithms}
\Crefname{myalgorithm}{Algorithm}{Algorithms}

\title{\bf An LP Algorithm for Counting Eulerian Orientations Through the Lens of Quasi-polymorphism}
 \author{
 \begin{tabular}{c@{\hspace{0.5cm}}c@{\hspace{0.5cm}}c}
 \normalsize Jincheng Guan &
\normalsize  Shuai Shao &
\normalsize  Ke Shi \\
 \normalsize  \url{guanjincheng@mail.ustc.edu.cn} & \normalsize \url{shao10@ustc.edu.cn}  & \normalsize \url{self.ke.shi@gmail.com}\\
 \end{tabular}
 \\[0.5cm]
\normalsize School of Computer Science and Technology \& Hefei National Laboratory\\
\normalsize  University of Science and Technology of China}
 \date{}
\begin{document}
\maketitle
\thispagestyle{empty}

\begin{abstract}
The weighted Eulerian orientation counting problem ($\ceo$) plays a key role in the complexity classification program for Holant problems.
A recent result~\cite{meng2025fpnp} established an $\pnp$ versus $\#\mathrm P$-hard dichotomy for $\ceo$ problems. The tractable side of this dichotomy can be characterized by functions admitting quasi-polymorphisms of the ternary XOR operation, leaving open whether these cases on the $\pnp$ side are in fact in FP.
In this paper, we settle this question by giving a polynomial-time algorithm for all cases on the $\pnp$ side. Consequently, we obtain a complete FP versus $\#\mathrm P$ dichotomy for counting weighted Eulerian orientations, and further for complex-valued Holant problems with an odd-arity signature.

Our algorithm is based on a linear programming relaxation, but we use it in a nonstandard way. Instead of proving that the relaxation is integral and solving the problem directly from an optimal LP solution, we use the relaxation as a structural tool to lift the quasi-polymorphism condition to an ordinary polymorphism condition. This reveals an affine local structure of the constraint functions, which leads to tractability.

\end{abstract}

\newpage

\section{Introduction}

Given an undirected graph $G$, an \emph{Eulerian orientation} of $G$ is an assignment of a direction to every edge such that, at each vertex, the number of incoming edges equals the number of outgoing edges. 
A graph admits an Eulerian orientation if and only if every vertex has even degree, by the classical Euler theorem. 
Thus, deciding whether a graph admits an Eulerian orientation is straightforward.
The corresponding counting problem is substantially more difficult. Mihail and Winkler showed that counting the Eulerian orientations of a given graph is \(\#\mathrm{P}\)-complete~\cite{mihail1992ontn}. 
However, the problem may become tractable when suitable local restrictions are imposed on the orientations of edges. An intriguing example arises from the partition function of the six-vertex model~\cite{pauling1935structure,slater1941theory,rys1963uber}, one of the most extensively studied models in statistical physics. In this example, the underlying graph is $4$-regular, and at each vertex one incident edge is restricted to be oriented toward the vertex. Counting Eulerian orientations satisfying these local restrictions can be done in polynomial time~\cite{cai2018complexity}. More generally, this type of restriction on the orientations of edges can be represented by constraint functions placed at the vertices.

Motivated by the question of which local constraint functions give rise to tractable counting problems, the framework of counting weighted Eulerian orientations, denoted by $\#\mathrm{EO}$, was introduced in~\cite{cai2020beyond}. In a $\#\mathrm{EO}$ instance, each vertex $v$ is assigned a constraint function $f_v$, also called a \emph{signature}, whose value determines the local contribution at $v$ given an orientation. 
Since each edge has two ends, an  orientation of an edge is  denoted by assigning 
$0$ to the head and $1$ to the tail.
Thus, for an edge $e=(u, v)$, its two possible orientations are represented by $(0, 1)$ or $(1, 0)$, and 
locally at every vertex $v$,  a
$0$ represents an incoming edge and a $1$ represents an outgoing edge. 
An Eulerian orientation corresponds to
an assignment to each edge ($01$ or $10$) where the numbers of 0's and 1's at
each $v$ are equal. 
Then, in a \#EO instance, the signatures are supported on inputs of Hamming weight exactly half of their arity.
A complexity dichotomy for this framework was first obtained under a physical symmetry known as \emph{arrow reversal symmetry} (ARS)~\cite{cai2020beyond}. 
Such a dichotomy classifies every problem in the framework as either polynomial-time solvable (tractable) or $\#\mathrm{P}$-hard, leaving no intermediate cases within the classification. 
In this paper, we remove the ARS assumption and consider $\#\mathrm{EO}$ problems defined by arbitrary complex-valued signatures. 
Our main result is a complete FP versus $\#\mathrm{P}$-hard dichotomy for this general setting.

The significance of the $\#\mathrm{EO}$ framework goes beyond the problem of counting Eulerian orientations itself. It plays a fundamental role in the complexity classification of counting problems for two reasons.
First, the $\#\mathrm{EO}$ framework  provides a bridge between counting constraint satisfaction problems ($\#\mathrm{CSP}$)~\cite{cai2014complexity} and Holant problems~\cite{cai2009holantpa}. On one hand, although  $\#\mathrm{EO}$ signatures are required to be supported on half-weight inputs, the  $\#\mathrm{EO}$ framework is expressive enough to encode arbitrary Boolean $\#\mathrm{CSP}$, including those whose constraint functions are not supported on half-weight inputs.
On the other hand, $\#\mathrm{EO}$ constitutes an indispensable component in the complexity classification of Holant problems. Holant is a substantially more expressive counting framework that encompasses $\#\mathrm{CSP}$, counting perfect matchings~\cite{cai2009holantpa}, partition functions of statistical-physics models such as the six-vertex and eight-vertex models~\cite{cai2023complexity}, and the classical simulation of quantum circuits~\cite{cai2018clifford}. The dichotomy for real-valued Holant problems relies crucially on the classification of $\#\mathrm{EO}$ under ARS~\cite{shao2020real}. Under a suitable holographic transformation, ARS signatures correspond precisely to real-valued signatures.
More recently, complex-valued Holant problems containing an odd-arity signature were classified as either solvable in \(\mathrm{FP}^{\mathrm{NP}}\) (i.e., in polynomial time with access to an NP oracle), or \(\#\mathrm{P}\)-hard~\cite{meng2025odd}.
This classification was based on a corresponding \(\mathrm{FP}^{\mathrm{NP}}\) versus \(\#\mathrm{P}\)-hard classification for complex-valued $\#\mathrm{EO}$ problems without assuming ARS~\cite{meng2025fpnp}. 
Consequently, obtaining a complete FP versus $\#\mathrm{P}$-hard dichotomy for complex-valued Holant problems requires first resolving the remaining \(\mathrm{FP}^{\mathrm{NP}}\) cases for $\#\mathrm{EO}$ problems. The dichotomy established in this paper provides exactly this missing step.

A major difficulty in classifying general $\#\mathrm{EO}$ problems is that, once ARS is removed, previously unknown and sometimes rather unexpected tractable classes begin to emerge. 
The first such signatures appeared in the complexity classification of the six-vertex model~\cite{cai2018complexity}. 
Because these signatures have arity four, they admit relatively explicit descriptions.
 However, the structure becomes considerably more complicated with higher arities.
A recent line of work generalized these arity-four examples and discovered new tractable classes of signatures for $\#\mathrm{EO}$ problems~\cite{meng2025p}. Remarkably, the classes discovered in~\cite{shao2026eulerian} can be characterized perfectly by Hadamard codes.
This establishes an unexpected connection between two classical mathematical concepts, Eulerian orientations and Hadamard codes, through counting complexity.
The appearance of such nontrivial tractable classes illustrates why a complete dichotomy of \#EO problems is challenging. 
Without an exhaustive classification, it is difficult to predict whether further exceptional tractable families remain undiscovered.

The final missing family emerged in \cite{meng2025fpnp}, denoted as  $\mitsuup$ and $\mitsudown$ signatures. 
It is proved in \cite{meng2025fpnp} that apart from these signatures, the \#EO problem is \#P-hard for all other complex-valued signatures.
However, \#EO problems defined by these $\mitsuup$ and $\mitsudown$ signatures exhibit a phenomenon that had not previously appeared in the complexity classification of counting problems: they can be solved in polynomial time with access to an NP oracle for the corresponding decision problems, which ask whether there exists an assignment with nonzero evaluation. This yields an $\mathrm{FP}^{\mathrm{NP}}$ versus $\#\mathrm{P}$-hard classification for general \#EO.
If one believes that $\mathrm{FP}^{\mathrm{NP}}\neq \#\mathrm{P}$, and that the $\#\mathrm{EO}$ framework does not contain natural \#P-intermediate problems, then these exceptional problems should be expected to lie in FP. However, proving this appears to be highly nontrivial,  and it   remained unresolved in~\cite{meng2025fpnp}. 

To convey how unusual the remaining tractable family is, we review the smallest nontrivial example, $f_{56}$,\footnote{
During our investigation of the special signature $f_{56}$, we used OpenAI's GPT-5.4 series models as a research aid.
}
among the $\mitsuup$ and $\mitsudown$ signatures.
The signature $f_{56}$ has arity $56$. Indeed, $56$ is the smallest possible arity.
It has a support set $\{\alpha^1,\alpha^2,\ldots,\alpha^5\}$ of size $5$, with $f_{56}(\alpha^i)=1$ for $i\in\{1,\ldots,5\}$ and $f_{56}(\alpha)=0$ otherwise.
Clearly, $f_{56}$ is completely determined by its support set.
Below, we list its five support inputs.
Following~\cite{meng2025p}, define the matrices
\[
H_0=
\begin{bmatrix}
0\\0\\0\\0\\0
\end{bmatrix},
\quad
H_2=
\begin{bmatrix}
1&1&1&1&0&0&0&0&0&0\\
1&0&0&0&1&1&1&0&0&0\\
0&1&0&0&1&0&0&1&1&0\\
0&0&1&0&0&1&0&1&0&1\\
0&0&0&1&0&0&1&0&1&1
\end{bmatrix},
\quad
H_4=
\begin{bmatrix}
0&1&1&1&1\\
1&0&1&1&1\\
1&1&0&1&1\\
1&1&1&0&1\\
1&1&1&1&0
\end{bmatrix}.
\]
Also, define a $5$-by-$56$ matrix $R_{56}:=[\,H_0\ H_2^{\to 4}\ H_4^{\to 3}\,]$, where $A^{\to k}$ denotes the concatenation of $k$ copies of $A$.
Then, each row of $R_{56}$ is a support input $\alpha^i$ of $f_{56}$. 
The signature $f_{56}$ has the following properties. 
Each row $\alpha^i$ has Hamming weight $28$: it contains four $1$'s in each of the four $H_2$ blocks and four $1$'s in each of the three $H_4$ blocks. Hence $\operatorname{wt}(\alpha^i)=28$ for every $i\in[5]$, so $f_{56}$ is an EO signature of arity $56$.  For any pairwise distinct $i,j,k\in[5]$, we have $\operatorname{wt}(\alpha^i\oplus\alpha^j\oplus\alpha^k)=30>28$.

At first sight, there is little reason to expect the counting problem defined by such a signature $f_{56}$ to admit a polynomial-time algorithm. This example illustrates one of the main values of a complexity dichotomy: an exhaustive classification can reveal tractable problems whose algebraic structure and algorithmic tractability would otherwise be very difficult to anticipate.

In this paper, we handle the \#EO problems defined by the signature $f_{56}$, as well as all other $\mitsuup$ and $\mitsudown$ signatures, by giving a unified polynomial-time algorithm.
This settles the tractability of all cases previously placed in the $\mathrm{FP}^{\mathrm{NP}}$ class.

\begin{theorem}\label{thm:main}
    Every $\#\mathrm{EO}$ problem belonging to the $\mathrm{FP}^{\mathrm{NP}}$ class of the classification in~\cite{meng2025fpnp} is solvable in polynomial time.
    Consequently, the previous $\mathrm{FP}^{\mathrm{NP}}$ versus $\#\mathrm{P}$-hard classification of \#EO is strengthened to a complete FP versus $\#\mathrm{P}$-hard dichotomy.
\end{theorem}

Since the classification of complex-valued $\#\mathrm{EO}$ problems serves as the central ingredient in the classification of complex-valued Holant problems containing an odd-arity signature, our result immediately yields the following consequence.

\begin{corollary}
    The previous $\mathrm{FP}^{\mathrm{NP}}$ versus $\#\mathrm{P}$-hard classification of complex-valued Holant problems containing an odd-arity signature is strengthened to a complete FP versus $\#\mathrm{P}$-hard dichotomy.
\end{corollary}

Our algorithm is based on a carefully designed linear-programming relaxation. We first give an overview of the main ideas.
We characterize the $\mitsuup$ and $\mitsudown$ signatures by a \emph{quasi-polymorphism} condition associated with the ternary XOR operation. Our goal is to show that, for a $\#\mathrm{EO}$ instance with $\mitsuup$ and $\mitsudown$ signatures, this quasi-polymorphism condition can be strengthened to an ordinary polymorphism condition. The latter has a much more rigid algebraic structure, which is precisely an affine linear subspace, and eventually leads to tractability. For this purpose, we formulate a linear system using indicator variables for local assignments in truth-table form, rather than the original edge variables of the $\#\mathrm{EO}$ instance.

Our use of the linear system, however, differs from its more standard role in linear-programming algorithms. We do not solve the linear-programming relaxation directly and output an integral point as a solution to the decision problem corresponding to the counting instance. Recall that this is why an NP oracle is needed. In fact, we cannot prove that the linear-programming relaxation has an integral feasible region. Rather, the linear-programming relaxation serves as a structural tool for lifting the quasi-polymorphism condition to an ordinary polymorphism condition. In addition, our analysis of the linear system is uncommon and may initially appear unnatural. Instead of analyzing the linear system itself, we study its cubic expressions, which interact precisely with the ternary XOR operation underlying the quasi-polymorphism condition.

Thus, the main technical contribution of this paper is not merely the removal of an NP oracle. It also provides a structural explanation for why the exceptional $\mathrm{FP}^{\mathrm{NP}}$ cases are tractable: the constraints imposed by a $\#\mathrm{EO}$ instance eliminate the freedom allowed by the quasi-polymorphism condition and force it to become a tractable polymorphism condition.

\section{Preliminaries}

\subsection{Counting Eulerian orientation}
A signature of arity $r$ is a function
$f:\{0,1\}^r\to\mathbb C$.
We write
$\mathrm{Var}(f)=(x_1,\ldots,x_r)$ for the variables of $f$ in a fixed order.
For an input
$\alpha=(\alpha_1,\ldots,\alpha_r)\in\{0,1\}^r$, the value of $f$ on
$\alpha$ is denoted by $f(\alpha)$.
We can also view $\alpha$ as a $0$-$1$ bit-string. In this notation,
 the $i$-th bit $\alpha_i$ of the string is
the value assigned to the $i$-th variable $x_i$ of $f$, which can also be written as $\alpha(x_i)$.
Note that the bit-string
notation is always understood with respect to the given order of the
variables.

\begin{definition}[Support set, $\mathcal{S}(f)$]
    The support set of a signature $f$ of arity $r$, denoted by $\su(f)$, is the set of all input strings on which $f$ takes non-zero values, i.e.,
$
    \su(f)
    =
    \{\alpha\in\{0,1\}^{r} \mid f(\alpha)\neq 0\}
$.
\end{definition}

For a positive integer $n$, let $[n]:=\{1,2,\ldots,n\}$.
Suppose that $|\su(f)|=m$. Fix an order of all strings in $\su(f)=\{\alpha^1, \alpha^2, \ldots, \alpha^m\}$.
We give the following truth-table representation $T(f)$ of $f$ which can be written as an $m\times r$ matrix where the $i$-th row is the support $\alpha^i$, and the columns are indexed by the variables $\mathrm{Var}(f)=(x_1, x_2, \ldots, x_r)$ of $f$ in a given order.
Then, $T_{i,j}(f)=\alpha^i(x_j)$.

We define the following useful sets. For a $0$-$1$ bit string $\alpha$ of even length $2d$, let $|\alpha|$ denote its length and let $\wt(\alpha)$ denote its Hamming weight, i.e., the number of 1's in $\alpha$. Define
$
  \eoe := \{\alpha\in\{0,1\}^* \mid \wt(\alpha)=\frac{1}{2}|\alpha|\},
  \eosg := \{\alpha\in\{0,1\}^* \mid \wt(\alpha)>\frac{1}{2}|\alpha|\},$ and $
  \eosl := \{\alpha\in\{0,1\}^* \mid \wt(\alpha)<\frac{1}{2}|\alpha|\}.
$
A signature $f$ is called an $\eo$ signature if $\su(f) \subseteq \eoe$.

For an undirected graph $G=(V,E)$, an orientation $\sigma$ of $G$ is an assignment of a direction to each edge of $G$. Under our convention,
a direction $u\rightarrow v$ of an edge $e=(u,v)$ can be encoded by assigning the value $1$ to the endpoint $u$ and the value $0$ to the endpoint
$v$. Moreover, an orientation $\sigma$ is called an Eulerian orientation if the number of incoming edges equals the number of outgoing
edges at every vertex, or equivalently, every vertex takes the same number of $0$ inputs and $1$ inputs. We
denote by $\mathrm{EO}(G)$ the set of all Eulerian orientations of $G$.

\begin{definition}[Counting EO]
    A counting complex-weighted Eulerian orientation problem is parametrized by a set of EO signatures $\mathcal{F}$, and is denoted by
    $\ceo(\mathcal{F})$. An instance of $\ceo(\mathcal{F})$ is a signature grid $\Omega=(G,\pi)$ over $\mathcal{F}$, which consists of a graph $G=(V,E)$ and a mapping $\pi$ that assigns to each vertex $v\in V$ an $f_v\in \mathcal{F}$ and a linear order of the incident edges at $v$. The output is
    $$
        \ceo_{\Omega}=\sum_{\sigma\in \mathrm{EO}(G)}\prod_{v\in V} f_v(\sigma|_{E(v)}),
    $$
    where $\sigma$ is an Eulerian orientation of $G$ and $\sigma|_{E(v)}$ is the orientation $\sigma$ restricted to the set $E(v)$ of incident edges of $v$.
\end{definition}

\begin{definition}[Effective orientations and decision EO]
    Given an $\ceo$ instance $\Omega=(G, \pi)$ where $G=(V, E)$ and $\pi(v)=f_v$, we say $\sigma\in \mathrm{EO}(G)$ is an effective Eulerian orientation of $\Omega$ if $\prod_{v\in V} f_v(\sigma|_{E(v)})\neq 0$.
    The Decision $\eo$ problem is to decide whether there exists an effective Eulerian orientation $\sigma\in \mathrm{EO}(G)$ for an instance $\Omega=(G, \pi)$.
\end{definition}

    \begin{definition}[Locally effective support set, $\mathcal{E}(f_v)$]
        Given an $\ceo$ instance $\Omega=(G, \pi)$ where $G=(V, E)$ and $\pi(v)=f_v$, 
    we say an input string $\alpha\in \su(f_v)$ is locally effective for $f_v$ if there exists an effective Eulerian orientation $\sigma\in \mathrm{EO}(G)$ of $\Omega$ such that $\alpha=\sigma|_{E(v)}$. 
    We use $\mathcal{E}(f_v)$ to denote the set of all locally effective strings for $f_v$, namely the \emph{locally effective support set}. 
    \end{definition}

    For a signature $f$ of arity $r$ and a set $\mathcal{T}\subseteq\{0,1\}^r$,
we write $f|_\mathcal{T}$ for the signature obtained by restricting the support set of $f$
to $\mathcal{T}$, namely $(f|_\mathcal{T})(\alpha)=f(\alpha)$ if $\alpha\in \mathcal{T}$ and
$(f|_\mathcal{T})(\alpha)=0$ otherwise.

\begin{lemma}\label{lem:ceo value equal to a small support set}
      Let $\Omega=(G, \pi)$ with $G=(V, E)$ and $\pi(v)=f_v$ be an $\ceo$ instance.
      For every $v \in V$, suppose that  $\mathcal{T}_v$  is a set satisfying $\mathcal{E}(f_v)\subseteq \mathcal{T}_v$.
      Consider another $\ceo$ instance $\Omega'=(G, \pi')$ on the same underlying graph $G$ where $\pi'(v)=f_v\mid _{\mathcal{T}_v}$.
     Then, we have $\ceo_\Omega=\ceo_{\Omega'}$.
\end{lemma}

\begin{proof}
      By the definition of Counting $\eo$, $\ceo_{\Omega}=\sum_{\sigma\in \eo(G)}\prod_{v\in V}f_v(\sigma|_{E(v)})$. We compare the two sums term by term. If $\sigma$ is an effective Eulerian orientation of $\Omega$, then for every $v\in V$, the string $\sigma|_{E(v)}$ is locally effective, since $\mathcal{E}(f_v)\subseteq \mathcal{T}_v$, we get $(f_v|_{\mathcal{T}_v})(\sigma|_{E(v)})=(f_v|_{\mathcal{E}(f_v)})(\sigma_{E(v)})$ for every $v\in V$, so the term of $\sigma$ is unchanged. If $\sigma$ is not an effective Eulerian orientation, then $\prod_{v\in V} f_v(\sigma|_{E(v)})=0$, thus $\prod_{v\in V} (f_v|_{\mathcal{E}(f_v)})(\sigma|_{E(v)})=0$ and $\prod_{v\in V} (f_v|_{\mathcal{T}(v)})(\sigma|_{E(v)})=0$. Hence $\ceo_\Omega=\ceo_{\Omega'}$.
\end{proof}

Below, we give a truth-table representation, namely truth-table indicator vector, for effective Eulerian orientations. We will define a linear system in terms of these vectors.  

\begin{definition}[Truth-table indicator vector]\label{def:tt-indicator}
Let $\Omega=(G, \pi)$ be a $\ceo$ instance where $G=(V, E)$ and $\pi(v)=f_v$ for every $v\in V$.
Suppose that $|\su(f_v)|=m_v$, and $\su(f_v)$ is listed in a fixed order.
For every effective Eulerian orientation $\sigma$ of $\Omega$,
we define the following 0-1 valued vector $\boldsymbol \lambda^\sigma=(\lambda^\sigma_{v,\alpha})_{\alpha\in\su(f_v),\, v\in V}\in \{0, 1\}^{\sum_{v\in V}m_v}$, where $\lambda^\sigma_{v,\alpha}=1$ if $\sigma|_{E(v)}=\alpha$ and $\lambda^\sigma_{v,\alpha}=0$ otherwise.
We say $\boldsymbol \lambda^\sigma$ is an effective truth-table indicator of the instance $\Omega$.

\end{definition}

\subsection{Quasi-polymorphism}

We use $\oplus_3:\{0, 1\}^3\rightarrow \{0, 1\}$ to denote the ternary XOR operation, i.e., $\oplus_3(a, b, c)=a\oplus b \oplus c$.
The notion of polymorphism~\cite{dblp:journals/jacm/jeavonscg97,dblp:journals/siamcomp/bulatovjk05} is used in the algebraic approach to the CSP.
It is usually defined for relations, i.e., $\{0, 1\}$-valued functions.
We extend it to complex-valued signatures by considering its support set. We introduce quasi-polymorphism as a key notion in our linear programming algorithm.

\begin{definition}[Polymorphism and quasi-polymorphism]
    For an $\eo$ signature $f$ of arity $2d$,
 we say $f$ admits $\oplus_3$
 \begin{enumerate}
 \item as a \emph{polymorphism}, denoted by $f\in \mathsf{Pol}(\oplus_3)$, if for all $\alpha,\beta,\gamma\in\su(f)$,
    $\alpha\oplus\beta\oplus\gamma\in \su(f)$,
     \item  as an \emph{up-polymorphism}, denoted by $f\in \mathsf{Pol}^\uparrow (\oplus_3)$, if for all $\alpha,\beta,\gamma\in\su(f)$,
    $\alpha\oplus\beta\oplus\gamma\in \su(f)\cup\eosg$;
    \item and as a \emph{down-polymorphism}, denoted by $f\in \mathsf{Pol}^\downarrow (\oplus_3)$, if for all $\alpha,\beta,\gamma\in\su(f)$,
    $\alpha\oplus\beta\oplus\gamma\in \su(f)\cup\eosl.$
 \end{enumerate}

  Both up and down polymorphisms are called quasi-polymorphisms.

  For a signature set $\mathcal{F}$, if for every $f\in \mathcal{F}$, we have $f\in \mathsf{Pol}^\uparrow (\oplus_3)$, then we say $\mathcal{F}$ admits $\oplus_3$ as an up-polymorphism, denoted by $\mathcal{F}\subseteq \mathsf{Pol}^\uparrow (\oplus_3)$.
  Symmetrically, we can define $\mathcal{F}\subseteq \mathsf{Pol}^\downarrow (\oplus_3)$. Similarly, we can define $\mathcal{F}\subseteq \POLMID.$

\end{definition}

\begin{remark}
 If a signature $f$ admits $\oplus_3$ as a polymorphism, then its support set $\su(f)$ is an \emph{affine linear subspace} over $\Ftwo$. 
 We simply say $\su(f)$ is affine, or $f$ has an affine support. 
 
 Signatures in $\mathsf{Pol}^\uparrow (\oplus_3)$ and $\mathsf{Pol}^\downarrow (\oplus_3)$ were originally introduced in~\cite{meng2025fpnp}, and denoted as $\mitsuup$ and $\mitsudown$ signatures respectively.
 In this paper, we view these signatures through the lens of quasi-polymorphisms, which is a one-sided relaxation of the standard definition of a polymorphism.
\end{remark}

The following lemma is proved in~\cite{meng2025fpnp} for \#EO instances with signatures admitting $\oplus_3$ as a quasi-polymorphism.

\begin{lemma}[\cite{meng2025fpnp}]\label{lem:effective support set is affine}
    Let $\Omega=(G, \pi)$ with $G=(V, E)$ and $\pi(v)=f_v$ be an instance of  $\ceo(\mathcal{F})$.  
 If $\mathcal{F}\subseteq \mathsf{Pol}^\uparrow (\oplus_3)$ or $\mathcal{F}\subseteq \mathsf{Pol}^\downarrow (\oplus_3)$,
    then for every $v\in V$, the local effective support set $\mathcal{E}(f_v)$ is affine. 
\end{lemma}

\subsection{Known results}

We recall several theorems that will be used in this paper: two classes of
polynomial-time solvable $\eo$ signatures, $\Aclass$ and $\Pclass$, and the
$\pnp$ versus $\#\mathrm{P}$-hard classification for
$\ceo$ \cite{meng2025fpnp}.

\begin{definition}[The classes $\Aclass$ and $\Pclass$]
We use $\Aclass$ and $\Pclass$ for the two standard tractable classes in the
$\#\mathrm{CSP}$ dichotomy of~\cite{cai2014complexity}.  A signature
$f:\{0,1\}^r\to\mathbb C$ belongs to $\Aclass$ if either $f\equiv 0$, or
there exist a nonzero constant $\lambda\in\mathbb C$, a matrix $A$ over
$\Ftwo$, and $0$-$1$-valued affine linear forms
$\ell_1,\ldots,\ell_m$ over $\Ftwo$ such that, for
$X=(x_1,\ldots,x_r,1)^T$,
\[
    f(x_1,\ldots,x_r)
    =
    \lambda\,\chi_{AX=0}\,
    i^{\ell_1(X)+\cdots+\ell_m(X)}.
\]
Here $\chi$ is a 0-1 indicator function such that $\chi_{AX=0}=1$ iff $AX=0$, while the exponent
$\ell_1(X)+\cdots+\ell_m(X)$ is an ordinary integer sum, equivalently read
modulo $4$.
The class $\Pclass$ consists of signatures expressible as products of unary
signatures, binary equality signatures $=_2$, and binary disequality
signatures $\neq_2$.
\end{definition}
\begin{remark}
If a signature $f$ is in  $\mathscr{A}$ or $\mathscr{P}$, then its support $\su(f)$ is affine. 
\end{remark}

\begin{theorem}[Tractable classes $\Aclass$ and $\Pclass$]\label{cor:tractable of EO}
    Let $\mathcal F$ be a set of $\eo$ signatures. If $\mathcal{F}\subseteq \Aclass$ or $\mathcal{F}\subseteq \Pclass$, then $\ceo(\mathcal{F})\in \mathrm{FP}$.
\end{theorem}

Given $2d$ variables $x_1, x_2, \ldots, x_{2d} \in \{0, 1\}$ in a fixed order, a perfect pairing $M$ of them is a partition of the $2d$ variables into $d$ pairs, namely
$M=\{\{x_{i_1},x_{i_2}\},\{x_{i_3},x_{i_4}\},\dots,\{x_{i_{2d-1}},x_{i_{2d}}\}\}.$ 
Fixing a perfect pairing $M$ of the $2d$ variables, we define $$\{0, 1\}^M:=\{(x_1, x_2, \ldots, x_{2d})\in\{0, 1\}^{2d}\mid x_{i_1}\neq x_{i_2}, x_{i_3}\neq x_{i_4}, \ldots, x_{i_{2d-1}}\neq x_{i_{2d}}\}.$$

\begin{lemma}[\cite{cai2020beyond}]\label{lem:pairwise opposite}
    If an $\eo$ signature $f$ has an affine support, i.e., $f\in \POLMID$, then there exists a perfect pairing $M$ of the variables of $f$ such that $\mathcal{S}(f)\subseteq \{0, 1\}^M$.
    Moreover, $M$ can be founded in time polynomial on the arity of $f$.
\end{lemma}

\begin{definition}[{The classes $\eom[\Aclass]$ and $\eom[\Pclass]$}]
We say that an {\rm EO} signature $f$ of arity $2d$ is in the class $\eom[\Aclass]$  if for any perfect pairing $M$ of the $2d$ variables of $f$, we have $f|_{\{0, 1\}^M}\in \Aclass$. 
We say that $f$ is in the class $\eom[\Pclass]$  if for any perfect pairing $M$ of the $2d$ variables of $f$, $f|_{\{0, 1\}^M}\in \Pclass$. 
\end{definition}

\begin{lemma}\label{lem:from affine to pairwise}
    
    If $f\in \eom[\Pclass]$ or $\eom[\Aclass]$, then for any affine linear subspace $\mathcal{T}\subseteq \mathcal{S}(f)$, there exists some $\mathcal{T}'$ satisfying $\mathcal{T}\subseteq \mathcal{T}'$
    such that $f\mid_\mathcal{T'}\in \mathscr{P}$ or $\mathscr{A}$ respectively.

    In particular, if $f\in \eom[\Aclass]$, then $f\mid_\mathcal{T'}$ can be simply replaced by $f\mid_\mathcal{T}$.
\end{lemma}

\begin{proof}
    Suppose that $f$ is an $\eo$ signature of arity $2d$. Since $f|_{\mathcal{T}}$ is an $\eo$ signature and has affine support, by~\Cref{lem:pairwise opposite}, there exists a perfect pairing $M$ of $2d$ variables of $f|_{\mathcal{T}}$ such that $\mathcal{T}\subseteq \{0,1\}^M$. Since $f\in \eo^{\Pclass}$ or $\eo^{\Aclass}$, so $f|_{\{0,1\}^M}\in \Pclass$ or $\Aclass$ respectively. 
    
    Now we consider the case that $f\in \eo^{\Aclass}$, by the above discussion, we have $f|_{\{0,1\}^M}\in \Aclass$. By the definition of $\Aclass$, suppose that $\mathrm{Var}(f)=\{x_1,x_2,\dots,x_{2d}\}$, then $f|_{\{0,1\}^M}(x_1,x_2,\dots,x_{2d})=\lambda\chi_{AX=0}i^{\ell_1(X)+\dots+\ell_{2d}(X)}$ where $X=(x_1,\dots,x_{2n},1)$ and $AX=0$ is an affine linear system over $\mathbb{F}_2$. Since $\mathcal{T}$ is an affine linear system over $\mathbb{F}_2$, there exist an matrix $B$ over $\mathbb{F}_2$ such that $\mathcal{T}=\{(x_1,x_2,\dots,x_{2d}):BX=0\}$. Hence, $f|_{\mathcal{T}}(x_1,\dots,x_{2d})=f|_{\{0,1\}^M}(x_1,\dots,x_{2d})\chi_{BX=0}=\lambda\chi_{AX=0}\chi_{BX=0}i^{\ell_1(X)+\dots+\ell_{2d}(X)}=\lambda\chi_{CX=0}i^{\ell_1(X)+\dots+\ell_{2d}(X)}\in \Aclass$, where $C=\begin{bsmallmatrix}
        A\\ 
        B
    \end{bsmallmatrix}$.
\end{proof}

\begin{remark}
    We give an example to illustrate that for the case $f\in \eom[\Pclass]$,  why we cannot directly get  $f\mid_\mathcal{T}\in \mathscr{P}$ for any affine set $\mathcal{T}\subseteq \mathcal{S}(f_v)$. Let $f(x_1,y_1,x_2,y_2,x_3,y_3)=\prod_{i=1}^3 u(x_i)\prod_{i=1}^3\neq_2(x_i,y_i)$, where $\neq_2(x_i,x_j)=\chi_{x_i\oplus y_i=1}$ is the binary disequality function, and  $u(x_i)$ is a unary function with $u(0)=1$ and $u(1)=3$. Let $\mathcal{T}=\{(x_1,y_1,x_2,y_2,x_3,y_3):x_1\oplus y_1=1,x_2\oplus y_2=1,x_3\oplus y_3=1,x_1\oplus x_2 \oplus x_3=1\}$, which clearly an affine linear subspace. Then,
    One can check that $f\in \eo^{\Pclass}$ and indeed $f\in \Pclass $, but $f|_{\mathcal{T}}\notin \mathscr{P}$. 
\end{remark}

\begin{theorem}[$\pnp$ versus $\#\mathrm{P}$-hard  classification for $\ceo$~\cite{meng2025fpnp}]\label{thm:MWX-hard}
Let $\mathcal F$ be a set of $\eo$ signatures. If $\mathcal{F}$ satisfies the following two conditions:

\begin{enumerate}
    \item $\mathcal F\subseteq \POLUP$ or $\mathcal F\subseteq \POLDOWN$,
    \item $\mathcal F\subseteq\eom[\Aclass]$ or $\mathcal F\subseteq\eom[\Pclass]$,
\end{enumerate}
then $\ceo(\mathcal F)\in\mathrm{FP}^{\mathrm{NP}}$.  Otherwise,
$\ceo(\mathcal F)$ is $\#\mathrm{P}$-hard.
\end{theorem}

We briefly review the proof idea for the $\mathrm{FP}^{\mathrm{NP}}$ side and explain why an NP oracle is needed.

Consider an instance $\Omega=(G,\pi)$ of $\ceo(\mathcal{F})$, where $G=(V,E)$ and $\pi(v)=f_v\in\mathcal{F}$. Since either $\mathcal{F}\subseteq\POLUP$ or $\mathcal{F}\subseteq\POLDOWN$, \Cref{lem:effective support set is affine} implies that $\mathcal{E}(f_v)$ is affine for every $v\in V$. Moreover, since either $\mathcal{F}\subseteq\eom[\Aclass]$ or $\mathcal{F}\subseteq\eom[\Pclass]$, \Cref{lem:from affine to pairwise} implies that, for every $v\in V$, there exists a set $\mathcal{T}'_v$ satisfying
$\mathcal{E}(f_v)\subseteq\mathcal{T}'_v\subseteq\mathcal{S}(f_v)$
such that either $f_v\mid_{\mathcal{T}'_v}\in\mathscr{A}$ for all $v\in V$, or $f_v\mid_{\mathcal{T}'_v}\in\mathscr{P}$ for all $v\in V$.

Now consider the instance $\Omega'=(G,\pi')$, where either
$\pi'(v)=f_v\mid_{\mathcal{T}'_v}\in\mathscr{A}$
for all $v\in V$, or $\pi'(v)\in\mathscr{P}$ for all $v\in V$. By~\Cref{lem:ceo value equal to a small support set}, we have
$\ceo_{\Omega}=\ceo_{\Omega'}.$
Clearly, the value of $\ceo_{\Omega'}$ can be computed in polynomial time.
However, the difficulty is that computing $\mathcal{E}(f_v)$ is not known to be polynomial-time solvable. In~\cite{meng2025fpnp}, $\mathcal{E}(f_v)$ is obtained by deciding, for each $\alpha\in\mathcal{S}(f_v)$, whether $\alpha$ is locally effective. This decision problem is in NP, which is why an NP oracle is needed.

\section{An LP-based  algorithm}

In this section, we introduce our linear programming (LP)-based algorithm. We show how we use LP as a structural tool for lifting the quasi-polymorphism condition to an ordinary polymorphism condition.
Remarkably, our algorithm does not compute $\mathcal{E}(f_v)$ and, in fact, has little to do with $\mathcal{E}(f_v)$, although $\mathcal{E}(f_v)$ was proved to be affine in~\cite{meng2025fpnp}. Instead, we give a polynomial-time algorithm to compute a possibly larger set $\mathcal{L}(f_v)$, called the LP-feasible support set (\Cref{def:LP feasible support set}) satisfying
$\mathcal{E}(f_v)\subseteq \mathcal{L}(f_v)\subseteq \mathcal{S}(f_v).$
We further prove that $\mathcal{L}(f_v)$ is affine by analyzing linear-programming constraints. This gives tractability.

Now, we define a linear system for an   instance $\Omega=(G,\pi)$ of $\ceo(\mathcal F)$ where $\mathcal{F}$ satisfies the two conditions in~\cref{thm:MWX-hard}.
We give some notations.
Suppose $G=(V, E)$. For every $v\in V$ with degree $2d_v$, let $f_v \in \mathcal{F}$ of arity $2d_v$ be the signature labeled on $v$, i.e., $\pi(v)=f_v$.
Let $e\in E$ be an edge incident to $v$ and $\alpha \in \su(f_v)$ be a support vector of $f_v$.
Recall that $\alpha(e)$ denotes the $0$-$1$ value of the support string $\alpha$ with respect to the input value corresponding to the edge $e$.

Now, we give the following linear system.

\begin{equation}\label{eq:linear-system}
\begin{aligned}
\textbf{Variables:}\quad
& \lambda_{v,\alpha} \geq 0
&& \text{for all } v\in V \text{ and }\alpha\in\su(f_v), \\[0.6ex]
\textbf{Vertex constraints:}\quad
& \sum_{\alpha\in\su(f_v)}\lambda_{v,\alpha}=1
&& \text{for every } v\in V, \\[0.6ex]
\textbf{Edge constraints:}\quad
& \sum_{\alpha\in\su(f_u)}\alpha(e)\,\lambda_{u,\alpha}
  +\sum_{\alpha\in\su(f_v)}\alpha(e)\,\lambda_{v,\alpha}=1
&& \text{for every } e=(u,v)\in E .
\end{aligned}
\end{equation}

We use $\boldsymbol \lambda=(\lambda_{v,\alpha})_{\alpha\in\su(f_v),\, v\in V}$ to denote the vector of variables, $\mathsf{L}(\Omega)$ to denote the above system, and $\mathsf{P}(\Omega)$ to denote the feasible region of  $\mathsf{L}(\Omega)$, which is a polytope. 
We give the following two key properties of $\mathsf{L}(\Omega)$.
First, $\mathsf{P}(\Omega)$ is a linear relaxation of the solution space in terms of truth-table indicator vectors (Definition~\ref{def:tt-indicator}) for the corresponding decision $\eo$ problem.
\begin{lemma}\label{lem:eff is solution}
    If $\boldsymbol \lambda^\sigma=(\lambda^\sigma_{v,\alpha})_{\alpha\in\su(f_v),\, v\in V}\in \{0, 1\}^{\sum_{v\in V}m_v}$ is an effective truth-table indicator of $\Omega$, which corresponds to an effective Eulerian orientation $\sigma$ of $\Omega$,
    then $\boldsymbol \lambda^\sigma \in\mathsf{P}(\Omega)$.
\end{lemma}

\begin{proof}
    Obviously, $\boldsymbol \lambda^\sigma$ is a vector of variables of $\mathsf{L}(\Omega)$ and $\boldsymbol \lambda^\sigma\geq 0$. By \Cref{def:tt-indicator}, for every vertex $v\in V$, there is only one string $\alpha\in \su(f_v)$ such that $\lambda^\sigma_{v,\alpha}=1$, and $\lambda^\sigma_{v,\beta}=0$ for all other $\beta\in\su(f_v)$ where $\beta\neq \alpha$. 
    Thus, $\boldsymbol \lambda^\sigma$ satisfies the vertex constraints of $\mathsf{L}(\Omega)$. 

   Since $\sigma$ is an effective Eulerian orientation of $\Omega$, for every edge $e=(u,v)\in E$, $\sigma$ assigns a value $1$ to one endpoint and a value $0$ to the other endpoint. 
Then,    
    for every $e=(u,v)\in E$, we have $\sum_{\alpha\in\su(f_u)}\alpha(e)\,\lambda^\sigma_{u,\alpha}
  +\sum_{\alpha\in\su(f_v)}\alpha(e)\,\lambda^\sigma_{v,\alpha}=\sigma|_{E(v)}(e)+\sigma|_{E(u)}(e)=1$. Thus, $\boldsymbol \lambda^\sigma$ also satisfies the edge constraints of $\mathsf{L}(\Omega)$. 
  Therefore,
 $\boldsymbol \lambda^\sigma\in \mathsf{P}(\Omega)$.
\end{proof}

\begin{definition}[LP feasible support set, $\mathcal{L}(f_v)$]\label{def:LP feasible support set}

  Given an $\ceo$ instance $\Omega=(G, \pi)$ where $G=(V, E)$ and $\pi(v)=f_v$, 
    we say an input  $\alpha\in \su(f_v)$ is LP feasible for $f_v$ if $\max_{\boldsymbol \lambda\in \mathsf{P}(\Omega)}\lambda_{v,\alpha}>0$.
      We use $\mathcal{L}(f_v)$ to denote the set of all LP feasible inputs for $f_v$, namely the \emph{LP feasible support set}. 
\end{definition}

\begin{lemma}
    \label{cor:LinEinS}
    $\mathcal{E}(f_v)\subseteq \mathcal{L}(f_v)\subseteq \su(f_v).$
\end{lemma}

\begin{proof}
    By the definition of $\mathcal{E}(f_v)$, for every vertex $v\in V$ and each locally effective string $\beta\in \mathcal{E}(f_v)$, there exists an effective truth-table indicator $\boldsymbol \lambda^\sigma=(\lambda^\sigma_{v,\alpha})_{\alpha\in\su(f_v),\, v\in V}\in \{0, 1\}^{\sum_{v\in V}m_v}$ of $\Omega$ such that $\lambda^\sigma_{v,\beta}=1>0$. And by \cref{lem:eff is solution}, we have $\boldsymbol \lambda^\sigma\in \mathsf{P}(\Omega)$, so that $\max_{\boldsymbol \lambda\in \mathsf{P}(\Omega)}\lambda_{v,\beta}\geq \lambda^\sigma_{v,\beta}=1>0$
    and we can get $\beta\in \mathcal{L}(f_v)$. Then $\mathcal{E}(f_v)\subseteq \mathcal{L}(f_v)$. By the definition of $\mathcal{L}(f_v)$,
    we have $\mathcal{L}(f_v)\subseteq\su(f_v)$. Finally, $\mathcal{E}(f_v)\subseteq \mathcal{L}(f_v)\subseteq \su(f_v).$
\end{proof}

Second, for every vertex $v\in V$, the LP feasible support set $\mathcal{L}(f_v)$ determined by $\mathsf{L}(\Omega)$ is an affine linear subspace of 
$\mathbb{F}_2^{2d_v}$.

\begin{theorem}[$\mathcal{L}(f_v)$ is affine]\label{the:lifted to affine}
   Let $\Omega=(G, \pi)$ with $G=(V, E)$ and $\pi(v)=f_v$ be an instance of  $\ceo(\mathcal{F})$.  
 If $\mathcal{F}\subseteq \mathsf{Pol}^\uparrow (\oplus_3)$ or $\mathcal{F}\subseteq \mathsf{Pol}^\downarrow (\oplus_3)$,
    then for every $v\in V$, the LP feasible support set $\mathcal{L}(f_v)$ is affine.

\end{theorem}

We give a proof of \Cref{the:lifted to affine} in~\Cref{sec:main proof}. With~\cref{the:lifted to affine} in hand, we are ready to give our algorithm.

\refstepcounter{myalgorithm}
\noindent\textbf{Algorithm~\themyalgorithm. LP-based algorithm.}
\label{alg:lp-feasibility}

Given an instance of $\ceo(\mathcal F)$ where $\mathcal F$ satisfies the
tractability conditions of \cref{thm:MWX-hard}.

\begin{algo}
\textsc{Input:} An instance $\Omega=(G,\pi)$ where $\pi(v)=f_v$ for every $v\in V(G)$\\
\textsc{Output:} The evaluation $\ceo_{\Omega}$\\[.25em]
\begin{tabular*}{\linewidth}{@{}r@{\hspace{.8em}}p{.70\linewidth}@{\extracolsep{\fill}}l@{}}
\algline{1}{\kw{if} $\mathsf{P}(\Omega)=\emptyset$ \kw{then} \kw{return} $0$}{\algcmt{feasibility check}}
\algline{2}{\kw{for each} vertex $v\in V(G)$ \kw{do}}{\algcmt{find $\mathcal{L}(f_v)$}}
\algline{3}{\algin \kw{for each} $\alpha\in \su(f_v)$ \kw{do}}{}
\algline{4}{\algin \algin Solve the linear programming problem $\max_{\boldsymbol\lambda\in \mathsf{P}(\Omega)}\lambda_{v,\alpha}$}{}
\algline{5}{\algin Determine $\mathcal{L}(f_v)$}{}
\algline{6}{\algin Determine $\mathcal{T}_v'$ such that $\mathcal{L}(f_v)\subseteq\mathcal{T}_v'$ and $f_v|_{\mathcal{T}_v'}\in \Pclass $ or $\Aclass$}{}
\algline{7}{Compute the value $\ceo_{\Omega'}$ where $\Omega'=(G,\pi')$ with $\pi'(v)=f_v|_{\mathcal{T}_v'}$}{\algcmt{tractable instance}}
\algline{8}{\kw{return} the value $\ceo_{\Omega'}$}{}
\end{tabular*}
\end{algo}

\begin{proposition}\label{prop:algorithm-correctness}
    The algorithm output is correct. The algorithm runs in time ${\rm poly}(n,s)$ where $n$ is the number of vertices, and $s=\max_{v\in V} |\su(f_v)|$.
\end{proposition}
\begin{proof}
    By~\cref{lem:ceo value equal to a small support set}, we can get $\ceo_{\Omega}=\ceo_{\Omega'}.$ 
    So the algorithm outputs the correct value.

    We analyze the time complexity.
    The linear system $\mathsf{L}(\Omega)$ has at most $O(n\cdot s)$ many variables, and $O(n^2)$ many constraint. 
    We need to solve the LP problem $\max_{\boldsymbol\lambda\in \mathsf{P}}\lambda_{v, \alpha}$ for every vertex $v\in V(G)$ and every $\alpha\in \mathcal{S}(f_v)$. 
    So, at most $O(n\cdot s)$ many LP problems. 
    Thus, we can compute all $\mathcal{L}(f_v)$ in $\mathrm{poly}(n,s)$  time. Moreover, $\mathcal T_v'$, namely the perfect pairing described in \Cref{lem:from affine to pairwise}, can be found in  $\mathrm{poly}(n)$ time.
    By~\cref{the:lifted to affine} and~\cref{lem:from affine to pairwise}, $\Omega'$ is an instance
    of $\ceo(\Aclass)$ or $\ceo(\Pclass)$.
    By~\cref{cor:tractable of EO}, we can compute $\ceo_{\Omega'}$ in  $\mathrm{poly}(n)$ time. Thus, the total running time of the algorithm is $\mathrm{poly}(n,s).$
\end{proof}

\begin{proof}[Proof of \Cref{thm:main}]
If the two stated conditions hold, then \cref{alg:lp-feasibility} computes every
instance of $\ceo(\mathcal F)$ in polynomial time by
\Cref{prop:algorithm-correctness}. If at least one condition fails, then
$\#\mathrm{P}$-hardness follows from
\Cref{thm:MWX-hard}.
\end{proof}

\section{Exploiting the quasi-polymorphism condition in a cubic form}\label{sec:main proof}

 We give the following reformulation of the above linear system $\mathsf{L}(\Omega)$. 
For a 0-1 bit string $\alpha=\alpha_1\alpha_2\ldots\alpha_r\in \{0, 1\}^r$, 
let $\widehat{\alpha}_i=2\alpha_i-1\in \{-1, 1\}$, and $$\Delta(\alpha)=\sum_{i=1}^r \widehat{\alpha}_i=\sum_{i=1}^r(2\alpha_i -1)=2\wt(\alpha)-r.$$ 
Note that $\Delta(\alpha)=0$ if $\alpha\in \eoe$, $\Delta(\alpha)>0$ if $\alpha\in \eosg$, and $\Delta(\alpha)<0$ if $\alpha\in \eosl$. 
In particular, let $f_v$ be an $\eo$ signature of arity $2d$ assigned on some vertex $v$, and let $E(v)$ be the set of edges incident to $v$, which correspond to the variables of $f_v$.
For a support string $\alpha\in \mathcal{S}(f_v)$, we have  $\Delta(\alpha)=0$ since $\alpha\in \eoe$.

We say $\widehat\alpha=\widehat\alpha_1\widehat\alpha_2\ldots\widehat\alpha_{r}\in \{-1, 1\}^{r}$  is the $\pm 1$ sign representation of $\alpha$. 
Under the representation, one can check that $\alpha \oplus \beta  \oplus \gamma=\delta$ is equivalent to $\widehat{\alpha} \cdot \widehat{\beta}\cdot \widehat{\gamma}=\widehat{\delta}$ where $\cdot$ refers to the bitwise product.
Thus, for any $\alpha, \beta, \gamma \in \{0, 1\}^r$, we have $\Delta(\alpha\oplus\beta\oplus\gamma)=\sum_{i=1}^r \widehat\alpha_i\cdot\widehat\beta_i\cdot \widehat\gamma_i.$
This equality will serve as the bridge connecting the quasi-polymorphism condition with our linear system.

Now, we reformulate the edge constraints in $\mathsf{L}(\Omega
)$ in terms of the sign representation as follows.

\begin{equation}\label{eq:sign edge constraint}
\begin{aligned}
&\sum_{\alpha\in\su(f_u)}\widehat\alpha(e)\,\lambda_{u,\alpha}+\sum_{\alpha\in\su(f_v)}\widehat\alpha(e)\,\lambda_{v,\alpha}\\
  =& \sum_{\alpha\in\su(f_u)}(2\alpha(e)-1)\,\lambda_{u,\alpha}
  +\sum_{\alpha\in\su(f_v)}(2\alpha(e)-1)\,\lambda_{v,\alpha}\\
  =& 2\left(\sum_{\alpha\in\su(f_u)}\alpha(e)\,\lambda_{u,\alpha}+\sum_{\alpha\in\su(f_v)}\alpha(e)\,\lambda_{v,\alpha}\right)-\sum_{\alpha\in\su(f_u)}\,\lambda_{u,\alpha}-\sum_{\alpha\in\su(f_v)}\,\lambda_{v,\alpha}\\
  =& 2-1-1=0.
\end{aligned}
\end{equation}

  For simplicity, we define 
  $p_{v,e}(\boldsymbol\lambda)
  \;:=\;
\sum_{\alpha\in\su(f_v)}\widehat\alpha(e)\,\lambda_{v,\alpha},$ and the edge constraints are
\begin{equation}\label{eq:simplicity edge constraint}
    p_{u,e}(\boldsymbol\lambda)+p_{v,e}(\boldsymbol\lambda)=0 \quad \text{for every } e=(u,v)\in E.
\end{equation}
Also, we define  $\Psi_v(\boldsymbol \lambda)
    :=\sum_{e\in E(v)} p_{v,e}^3(\boldsymbol{\lambda})$.

A key, albeit somewhat unnatural, idea in our proof is to lift the edge constraint
$p_{u,e}(\boldsymbol{\lambda})+p_{v,e}(\boldsymbol{\lambda})=0$
to its cubic form
$p_{u,e}^3(\boldsymbol{\lambda})+p_{v,e}^3(\boldsymbol{\lambda})=0.$
By carefully analyzing all cubic terms in
\(\Psi_v(\boldsymbol{\lambda})\), we establish~\Cref{the:lifted to affine}.

\begin{lemma}\label{cor:Psi>=0}
  Suppose that $\boldsymbol{\lambda}\in \mathsf{P}(\Omega)$ for some instance $\Omega=(G, \pi)$ of $\ceo(\mathcal{F})$ where $G=(V, E)$ and $\pi(v)=f_v\in \mathcal{F}$. 
  Then,
    $$\Psi_v(\boldsymbol \lambda)=\sum_{\alpha,\beta,\gamma\in\su(f_v)}\lambda_{v,\alpha}\lambda_{v,\beta}\lambda_{v,\gamma}\Delta(\alpha\oplus \beta \oplus \gamma).$$
   Moreover, for all $v\in V$ and $\alpha, \beta, \gamma \in \su(f_v)$, we have $\Psi_v(\boldsymbol \lambda)\geq 0$ and $\lambda_{v,\alpha}\lambda_{v,\beta}\lambda_{v,\gamma}\Delta(\alpha\oplus \beta \oplus \gamma)\geq 0$  if $\mathcal{F}\subseteq \POLUP$;  $\Psi_v(\boldsymbol \lambda)\leq 0$ and $\lambda_{v,\alpha}\lambda_{v,\beta}\lambda_{v,\gamma}\Delta(\alpha\oplus \beta \oplus \gamma)\leq 0$  if $\mathcal{F}\subseteq \POLDOWN$.
\end{lemma}

\begin{proof}
    By direct calculation, we have 
    \begin{equation*}
    \begin{aligned}
        \Psi_v(\boldsymbol \lambda)
    =\sum_{e\in E(v)} p_{v,e}^3&=\sum_{e\in E(v)}\left(\sum_{\alpha\in\su(f_v)}\widehat\alpha(e)\,\lambda_{v,\alpha}\right)^3\\
    &=\sum_{e\in E(v)}\left(\sum_{\alpha,\beta,\gamma\in\su(f_v)}\widehat\alpha(e)\,\lambda_{v,\alpha}\widehat\beta(e)\,\lambda_{v,\beta}\widehat\gamma(e)\,\lambda_{v,\gamma}\right)\\
    &=\sum_{\alpha,\beta,\gamma\in\su(f_v)}\lambda_{v,\alpha}\lambda_{v,\beta}\lambda_{v,\gamma}\left(\sum_{e\in E(v)}\widehat\alpha(e)\widehat\beta(e)\widehat\gamma(e)\right)\\
    &=\sum_{\alpha,\beta,\gamma\in\su(f_v)}\lambda_{v,\alpha}\lambda_{v,\beta}\lambda_{v,\gamma}\Delta(\alpha\oplus \beta \oplus \gamma).
    \end{aligned}
\end{equation*}

If $\mathcal{F}\subseteq \POLUP$, then $\Delta(\alpha\oplus\beta\oplus\gamma)\geq 0$ for all $v\in V$ and $\alpha,\beta,\gamma\in \mathcal{S}(f_v)$. Recall that $\lambda_{v,\alpha}\geq0$ for all $v\in V$ and $\alpha\in \mathcal{S}(f_v)$, we have $\lambda_{v,\alpha}\lambda_{v,\beta}\lambda_{v,\gamma}\Delta(\alpha\oplus\beta\oplus\gamma)\geq 0$ and $\Psi_v(\boldsymbol{\lambda})\geq 0$. The case that $\mathcal{F}\subseteq \POLDOWN$ is symmetric.

\end{proof}

Combining~\Cref{cor:Psi>=0} and the edge constraints, we have the following corollary.

\begin{corollary}\label{cor:lambdaDelta=0}
     Suppose that $\boldsymbol{\lambda}\in \mathsf{P}(\Omega)$ for some instance $\Omega=(G, \pi)$ of $\ceo(\mathcal{F})$ where $G=(V, E)$, $\pi(v)=f_v\in \mathcal{F}$.
     If $\mathcal{F}\subseteq \POLUP$ or $\mathcal{F}\subseteq \POLDOWN$, then for all $v\in V$ and all $\alpha, \beta, \gamma \in \su(f_v)$ which are not necessarily distinct, we have $\Psi_v(\boldsymbol \lambda)= 0$ and $\lambda_{v,\alpha}\lambda_{v,\beta}\lambda_{v,\gamma}\Delta(\alpha\oplus \beta \oplus \gamma)= 0$.
\end{corollary}

\begin{proof}
    We consider the case that $\mathcal{F}\subseteq \POLUP$; the down case is symmetric.
    Summing the cubic form of edge constraints over all edges $e=(u,v)\in E$, we obtain
    \begin{align*}
        0&= \sum_{e=(u,v)\in E}p^3_{u,e}(\boldsymbol\lambda)+p^3_{v,e}(\boldsymbol\lambda)=\sum_{v\in V}\sum_{e\in E(v)}p^3_{v,e}(\boldsymbol\lambda)=\sum_{v\in V} \Psi_v(\boldsymbol \lambda).
    \end{align*}
    By~\Cref{cor:Psi>=0}, for every $v\in V$, we have $\Psi_v(\boldsymbol \lambda)\geq 0$. Thus $\Psi_v(\boldsymbol \lambda)= 0$ for every $v\in V$.
    Similarly, we have $\lambda_{v,\alpha}\lambda_{v,\beta}\lambda_{v,\gamma}\Delta(\alpha\oplus \beta \oplus \gamma)= 0$ for all $v\in V$ and all $\alpha, \beta, \gamma \in \su(f_v)$.
\end{proof}

Below, we will prove $\mathcal{L}(f_v)$ is affine.

\begin{lemma}\label{lem:XOR in S(f_v)}
        Let $\Omega=(G, \pi)$ be an instance of $\ceo(\mathcal{F})$ where $G=(V, E)$, $\pi(v)=f_v\in \mathcal{F}$.
       Suppose that $\mathcal{F}\subseteq \POLUP$ or $\mathcal{F}\subseteq \POLDOWN$.
       If $\alpha, \beta, \gamma \in \mathcal{L}(f_v)$, then $\alpha\oplus \beta \oplus \gamma \in \mathcal{S}(f_v)$.
\end{lemma}

\begin{proof}
  Since $\alpha, \beta, \gamma \in \mathcal{L}(f_v)$,  by the definition
of $\mathcal{L}(f_v)$,  there exist  $\boldsymbol\lambda^{(1)},\boldsymbol\lambda^{(2)},\boldsymbol\lambda^{(3)}\in \mathsf{P}(\Omega)$ such that 
$\lambda^{(1)}_{v,\alpha}>0$, $\lambda^{(2)}_{v,\beta}>0$, $\lambda^{(3)}_{v,\gamma}>0$, respectively. Clearly, their average
$\boldsymbol\lambda^*:=\frac13(\boldsymbol\lambda^{(1)}+\boldsymbol\lambda^{(2)}+\boldsymbol\lambda^{(3)})\in \mathsf{P}(\Omega)$ since $\mathsf{P}(\Omega)$ is a polytope, and $\lambda^*_{v,\alpha},\lambda^*_{v,\beta},\lambda^*_{v,\gamma}>0$. 
By~\Cref{cor:lambdaDelta=0}, $\lambda^*_{v,\alpha}\lambda^*_{v,\beta}\lambda^*_{v,\gamma}\Delta(\alpha\oplus \beta \oplus \gamma)= 0$.
This forces $\Delta(\alpha\oplus \beta \oplus \gamma)= 0$, which implies $\alpha\oplus \beta \oplus \gamma \in \eoe$.
Note that $f_v\in \POLUP$ or $f_v\in \POLDOWN$, in both cases, $\alpha, \beta, \gamma \in \mathcal{L}(f_v)\subseteq \mathcal{S}(f_v)$ and $\alpha\oplus \beta \oplus \gamma \in \eoe$ together imply $\alpha\oplus \beta \oplus \gamma \in\mathcal{S}(f_v)$.
\end{proof}

For the particular signature $f_{56}$, \Cref{lem:XOR in S(f_v)} suffices to prove that $\mathcal{L}(f_v)$ is affine. 

\begin{lemma}[A restatement of \cref{the:lifted to affine} for the signature $f_{56}$]
       Let $\Omega=(G, \pi)$ be an instance of $\ceo(\{f_{56}\})$ where $G=(V, E)$, $\pi(v)=f_v=f_{56}$.
       Then, for every $v\in V$, $|\mathcal{L}(f_v)|\leq 2$ which clearly implies $\mathcal{L}(f_v)$ is affine. 
\end{lemma}

\begin{proof}
Consider the support set $\mathcal{S}(f_{56})$.
For any $\alpha, \beta, \gamma \in \mathcal{S}(f_{56})$ where $\alpha, \beta, \gamma$ are all distinct, one can check that
$\wt(\alpha \oplus \beta \oplus \gamma)=30>28$,
 thus $\alpha \oplus \beta \oplus \gamma \notin \mathcal{S}(f_v)$.

Suppose, for contradiction, that $\mathcal{L}(f_v)$ contains three distinct strings $\alpha,\beta,\gamma$. Since $f_{56}\in \POLUP$, by~\cref{lem:XOR in S(f_v)}, we have $\alpha\oplus\beta\oplus\gamma\in \mathcal{S}(f_v)$, contradicting the preceding paragraph. Thus
$|\mathcal{L}(f_v)|\leq 2$. A nonempty subset of $\{0,1\}^{56}$ of size at
most two is closed under the coordinatewise ternary XOR operation, hence affine.
\end{proof}

For general signature sets $\mathcal{F}$ admitting $\oplus_3$ as a quasi-polymorphism, we need to work a bit more.

\begin{theorem}[A restatement of~\Cref{the:lifted to affine}]
    Let $\Omega=(G,\pi)$ be an instance of $\ceo(\mathcal{F})$ where $G=(V,E)$, $\pi(v)=f_v\in \mathcal{F}$. Suppose that $\mathcal{F}\subseteq \POLUP$ or $\mathcal{F}\subseteq \POLDOWN$. Then, for every $v\in V$, $\mathcal{L}(f_v)$ is affine.
\end{theorem}
\begin{proof}
    We treat the case where $\mathcal{F} \subseteq \POLUP$; the down case is symmetric.
    Now we prove that if $\alpha, \beta, \gamma \in \mathcal{L}(f_v)$, then $\delta:=\alpha\oplus\beta \oplus \gamma \in \mathcal{L}(f_v)$.
    We may assume that $\alpha, \beta, \gamma$ are all distinct. Otherwise, it is trivial.

    By the definition
of $\mathcal{L}(f_v)$,  there exist  $\boldsymbol\lambda^{(1)},\boldsymbol\lambda^{(2)},\boldsymbol\lambda^{(3)}\in \mathsf{P}(\Omega)$ such that 
$\lambda^{(1)}_{v,\alpha}>0$, $\lambda^{(2)}_{v,\beta}>0$, $\lambda^{(3)}_{v,\gamma}>0$, respectively. Still, we consider the average
$\boldsymbol\lambda^*:=\frac13(\boldsymbol\lambda^{(1)}+\boldsymbol\lambda^{(2)}+\boldsymbol\lambda^{(3)})\in \mathsf{P}(\Omega)$, and $\lambda^*_{v,\alpha},\lambda^*_{v,\beta},\lambda^*_{v,\gamma}>0$. 
We will construct a vector $\boldsymbol \mu \in \mathsf{P}(\Omega)$ such that $\mu_{v, \delta}>0$.
This gives $\delta \in \mathcal{L}(f_v)$.

Construct a new vector
$
    \boldsymbol \mu
    =
    (\mu_{v,\zeta})_{v\in V,\ \zeta\in \mathcal{S}(f_v)}
$
by defining each coordinate as follows:
\[
    \mu_{v,\zeta}
    =
    \sum_{\substack{
    \alpha',\beta',\gamma'\in \mathcal{S}(f_v)\\
    \alpha'\oplus\beta'\oplus\gamma'=\zeta
    }}
    \lambda_{v,\alpha'}^{\ast}
    \lambda_{v,\beta'}^{\ast}
    \lambda_{v,\gamma'}^{\ast}.
\]
By the definition of $\mu_{v,\zeta}$, we have $\mu_{v,\zeta}\geq 0$ for every $v\in V$ and every $\zeta\in\mathcal{S}(f_v)$.

Now, we check that $\boldsymbol \mu$ satisfies both the vertex constraints and edge constraints of $\mathsf{L}(\Omega)$, which implies that $\boldsymbol \mu \in \mathsf{P}(\Omega)$.

By~\cref{cor:lambdaDelta=0}, since $\boldsymbol{\lambda}^*\in \mathsf{P}(\Omega)$, for every $v\in V$ and all $\alpha', \beta', \gamma' \in \su(f_v)$, we have
$\lambda^*_{v,\alpha'}\lambda^*_{v,\beta'}\lambda^*_{v,\gamma'}\Delta(\alpha'\oplus \beta' \oplus \gamma')= 0$.
If
$
    \alpha'\oplus\beta'\oplus\gamma'\notin\mathcal{S}(f_v)
$, then $\alpha'\oplus \beta' \oplus \gamma' \in {\rm HW}^>$ since $f_v\in \POLUP$.
Hence, $\Delta(\alpha'\oplus \beta' \oplus \gamma')>0$ and this implies $
    \lambda_{v,\alpha'}^{*}
    \lambda_{v,\beta'}^{*}
    \lambda_{v,\gamma'}^{*}
    =0$.

Therefore, for every $v\in V$, we have
\begin{align*}
    \sum_{\zeta\in \mathcal{S}(f_v)}\mu_{v,\zeta}
    &=
    \sum_{\zeta\in \mathcal{S}(f_v)}
    \sum_{\substack{
    \alpha',\beta',\gamma'\in \mathcal{S}(f_v)\\
    \alpha'\oplus\beta'\oplus\gamma'=\zeta
    }}
    \lambda_{v,\alpha'}^{*}
    \lambda_{v,\beta'}^{*}
    \lambda_{v,\gamma'}^{*}\\
    &=
    \sum_{\alpha',\beta',\gamma'\in \mathcal{S}(f_v)}
    \lambda_{v,\alpha'}^{*}
    \lambda_{v,\beta'}^{*}
    \lambda_{v,\gamma'}^{*}\\
    &=
    \left(\sum_{\alpha'\in\mathcal{S}(f_v)}
    \lambda_{v,\alpha'}^{*}\right)\left(\sum_{\beta'\in\mathcal{S}(f_v)}
    \lambda_{v,\beta'}^{*}\right)\left(\sum_{\gamma'\in\mathcal{S}(f_v)}
    \lambda_{v,\gamma'}^{*}\right)\\
    &=
    1.
\end{align*}
Thus $\boldsymbol\mu$ satisfies all the vertex constraints.

Also, by~\cref{cor:lambdaDelta=0}, for every $v\in V$, if
$
    \alpha'\oplus\beta'\oplus\gamma'\notin\mathcal{S}(f_v)
$
then $
    \lambda_{v,\alpha'}^{*}
    \lambda_{v,\beta'}^{*}
    \lambda_{v,\gamma'}^{*}
    =0$. For every edge $e=(u,v)\in E$, we have
\begin{align*}
    p_{v,e}(\boldsymbol\mu)
    &=
    \sum_{\zeta\in \mathcal{S}(f_v)}
    \widehat{\zeta}(e)\mu_{v,\zeta}\\
    &=
    \sum_{\zeta\in \mathcal{S}(f_v)}
    \widehat{\zeta}(e)
    \sum_{\substack{
    \alpha',\beta',\gamma'\in \mathcal{S}(f_v)\\
    \alpha'\oplus\beta'\oplus\gamma'=\zeta
    }}
    \lambda_{v,\alpha'}^{*}
    \lambda_{v,\beta'}^{*}
    \lambda_{v,\gamma'}^{*}\\
    &=
    \sum_{\alpha',\beta',\gamma'\in \mathcal{S}(f_v)}
    \widehat{\;\alpha'\oplus\beta'\oplus\gamma'\;}(e)
    \lambda_{v,\alpha'}^{*}
    \lambda_{v,\beta'}^{*}
    \lambda_{v,\gamma'}^{*}\\
    &=
    \sum_{\alpha',\beta',\gamma'\in \mathcal{S}(f_v)}
    \widehat{\alpha'}(e)
    \widehat{\beta'}(e)
    \widehat{\gamma'}(e)
    \lambda_{v,\alpha'}^{*}
    \lambda_{v,\beta'}^{*}
    \lambda_{v,\gamma'}^{*}\\
    &=
    \left(
    \sum_{\alpha'\in \mathcal{S}(f_v)}
    \widehat{\alpha'}(e)
    \lambda_{v,\alpha'}^{*}
    \right)
    \left(
    \sum_{\beta'\in \mathcal{S}(f_v)}
    \widehat{\beta'}(e)
    \lambda_{v,\beta'}^{*}
    \right)
    \left(
    \sum_{\gamma'\in \mathcal{S}(f_v)}
    \widehat{\gamma'}(e)
    \lambda_{v,\gamma'}^{*}
    \right)\\
    &=
    p_{v,e}^3(\boldsymbol\lambda^{*}).
\end{align*}
Similarly,
$
    p_{u,e}(\boldsymbol\mu)
    =
    p_{u,e}^3(\boldsymbol\lambda^{*}).
$

Since $\boldsymbol{\lambda}^\ast \in \mathsf{P}(\Omega)$, it satisfies the edge constraint $p_{v,e}(\boldsymbol\lambda^{*})
    +p_{u,e}(\boldsymbol\lambda^{*})=0$ and also its cubic form $p_{v,e}^3(\boldsymbol\lambda^{*})
    +p_{u,e}^3(\boldsymbol\lambda^{*})=0$ for every edge $e=(u, v)$.
Thus, for every edge $e=(u, v)$, we have 
\[
    p_{v,e}(\boldsymbol\mu)+p_{u,e}(\boldsymbol\mu)=p_{v,e}^3(\boldsymbol\lambda^{*})
    +p_{u,e}^3(\boldsymbol\lambda^{*})=0.
\]
So $\boldsymbol\mu$ satisfies all the edge constraints. Therefore, $\boldsymbol\mu$ is a feasible point of $\mathsf{P}(\Omega)$. 

Moreover, since
$
    \mu_{v,\delta}
    \geq
    \lambda_{v,\alpha}^{*}
    \lambda_{v,\beta}^{*}
    \lambda_{v,\gamma}^{*}
    >0,
$
we have $\delta\in \mathcal{L}(f_v)$ by the definition of $\mathcal{L}(f_v)$.
\end{proof}

\appendix
\section{Acknowledgment}
We would like to thank Boning Meng, Zhuxiao Tang, Juqiu Wang, and Mingji Xia for many valuable discussions. 
 During the investigation of the special signature $f_{56}$, we used OpenAI's GPT-5.4 series models as a
  research aid.  The materials are available
  at \url{https://github.com/KeShih/EO}.

\bibliography{main}
\bibliographystyle{alpha}

\end{document}